\documentclass[aps,prx,10pt,twocolumn,superscriptaddress,showpacs,footinbib,floatfix,subfigure,longbibliography,amsmath,amssymb]{revtex4-1}
\usepackage{graphicx}
\usepackage{hyperref}
\usepackage[normalem]{ulem}

\newtheorem{theorem}{Theorem}[section]
\newtheorem{lemma}[theorem]{Lemma}

\newenvironment{proof}[1][Proof]{\begin{trivlist}
\item[\hskip \labelsep {\bfseries #1}]}{\end{trivlist}}

\newcommand{\qed}{\nobreak \ifvmode \relax \else
      \ifdim\lastskip<1.5em \hskip-\lastskip
      \hskip1.5em plus0em minus0.5em \fi \nobreak
      \vrule height0.75em width0.5em depth0.25em\fi}
      
\newcommand{\ci}{\cos\frac{\theta_i}{2}}
\newcommand{\si}{\sin\frac{\theta_i}{2}}

\newcommand{\cf}{\cos\frac{\theta_f}{2}}
\renewcommand{\sf}{\sin\frac{\theta_f}{2}}

\begin{document}
\title{When amplification with weak values fails to suppress technical noise}
\author{George~C.~Knee}
\email{george.knee@materials.ox.ac.uk}
\affiliation{Department of Materials, University of Oxford, Oxford OX1 3PH, United Kingdom}
\author{Erik~M.~Gauger}
\affiliation{Centre for Quantum Technologies, National University of Singapore, 3 Science Drive 2, Singapore 117543}
\affiliation{Department of Materials, University of Oxford, Oxford OX1 3PH, United Kingdom}
\date{\today}    
\begin{abstract}\noindent
The application of postselection to a weak quantum measurement leads to the phenomenon of weak values. Expressed in units of the measurement strength, the displacement of a quantum coherent measuring device is ordinarily bounded by the eigenspectrum of the measured observable. Postselection can enable an interference effect that moves the average displacement far outside this range, bringing practical benefits in certain situations. Employing the Fisher information metric, we argue that the amplified displacement offers no fundamental metrological advantage, due to the necessarily reduced probability of success. Our understanding of metrological advantage is the possibility of a lower uncertainty in the estimate of an unknown parameter with a large number of trials. We analyze a situation in which the detector is pixelated with a finite resolution, and in which the detector is afflicted by random displacements: imperfections which degrade the fundamental limits of parameter estimation. Surprisingly, weak-value amplification is no more robust to them than a technique making no use of the amplification effect brought about by a final, postselected measurement.
\end{abstract}
\maketitle
\section{Introduction}
In recent years the Aharonov, Albert, and Vaidman (AAV) effect~\cite{AharonovAlbertCasher1987,AharonovAlbertVaidman1988,DuckStevensonSudarshan1989} has received much attention. The phenomenon arises through a combination of both i) a weak quantum measurement of an \emph{initial state} of a `system' by a quantum coherent `meter'  and ii) judicious postselection of the system into an unlikely \emph{final state}. A salient feature of the phenomenon is known as `weak-value amplification' (WVA): one may, by fine control of experimental parameters, arrange for an anomalously large \emph{average} displacement of the meter wavefunction -- but this only occurs infrequently.

The pre- and postselected ensemble is characteristic of the `two state-vector approach',  which attempts to restore time symmetry to quantum mechanics~\cite{AharonovBergmannLebowitz1964}. The AAV effect may thus be of fundamental importance, although it is not exclusively a quantum one: especially in optical scenarios it can be thought of classically~\cite{AielloWoerdman2008}. Applications of AAV's work continue to excite the quantum information field: for example in the estimation of unknown quantum states, in testing Leggett-Garg inequalities, and in testing uncertainty and complementarity relations~\cite{LundeenSutherlandPatel2011,LundeenBamber2012,SalvailAgnewJohnson2013,RozemaDarabiMahler2012,DresselBroadbentHowell2011,WestonHallPalsson2013}. 

Here we concentrate on the foremost technological application of the AAV effect: in the estimation of the small constant (known as an interaction parameter) found to couple a suitably defined system and meter.  The meter wavefunction encodes information about the system through their mutual interaction. It is often claimed that a larger-than-usual displacement of the wavefunction will be of use when technological constraints limit one's ability to detect very small interaction parameters. Experiments have been performed that suggest exactly this: the most prominent claim a suppression of technical issues such as pointing stability, detector saturation, finite resolution, unwanted displacements and noise with long correlation times~\cite{RitchieStoryHulet1991,HostenKwiat2008,DixonStarlingJordan2009,BrunnerSimon2010,GorodetskiBliokhStein2012}. In this work we show that, for the most ostensible technical issues, namely random transverse displacements and pixelation, such an advantage only prevails whilst sub-optimal estimation procedures are employed. If the fundamental limit on the uncertainty in the estimate of the interaction parameter is attainable (by optimal estimation), the advantage disappears. We consider the performance of a standard technique which matches or betters WVA by using the same measurement strength, but without using postselection. 

Previously, we have shown that weak-value amplification techniques offer no advantage for a parameter estimation problem with a qubit meter, even under decoherence~\cite{KneeBriggsBenjamin2013}. By contrast, here we consider the canonical WVA scenario with a continuous degree of freedom acting as the meter. Importantly, we include a treatment of technical imperfections. This is the most significant difference between our work and a recent paper by Tanaka and Yamamoto~\cite{TanakaYamamoto2013}, who discuss the role of postselection in WVA metrology.

In Section \ref{model} of this article we introduce our physical model, and motivate the use of postselection to enable weak-value amplification. In Section \ref{metric} we introduce the Fisher information metric and maximum-likelihood estimation: here we also define a standard approach to parameter estimation which makes no use of postselection. Section \ref{ideal detector} examines a noise-free scenario, comparing WVA to our standard approach. In Section \ref{failedpostselection} we consider improvements to weak-value amplification by making full use of any discarded data. In Sections \ref{pixelation} and \ref{jitter} we allow detection to suffer from pixelation and jitter, respectively. For each imperfection we include worked examples with Gaussian meter wavefunctions, and allow for both real and imaginary weak-values. Sections \ref{exactcalculations} and \ref{exactcalculations_failedpostselection} treat the situation without making approximations. A short discussion follows in Section \ref{discussion}. 
\begin{figure}
\includegraphics[width=\linewidth]{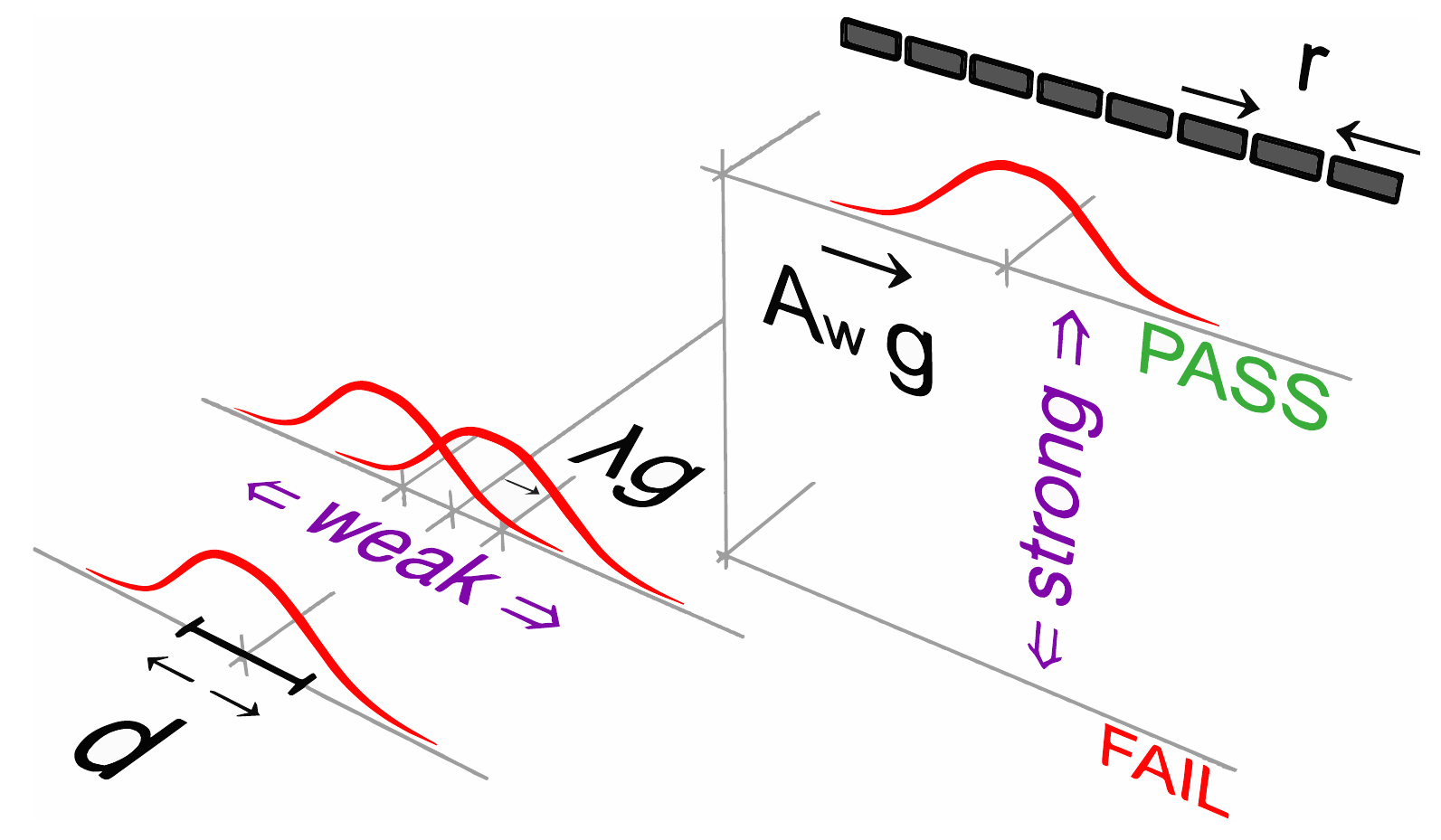}
\caption{\label{AAVsetup}The AAV setup. A particle's spatial wavefunction is weakly coupled to an internal observable $\mathbf{A}$, causing it to undergo a small lateral shift conditional on the internal state. It is then strongly post-selected into the eigenstate of another observable $\mathbf{B}$ and allowed to impinge on a array of pixels with width $r$. By tuning certain parameters one can induce larger-than-usual average displacements at the detection stage. The Figure illustrates the $\textrm{dim}(\mathcal{H})=2$ case.}
\end{figure}
\section{Model}
\label{model}
We imagine a scenario as depicted in Fig.~\ref{AAVsetup}. Begin with a particle described by a product of its internal state (defined on a finite-dimensional Hilbert space $\mathcal{H}$) and its transverse spatial wavefunction (a continuous degree of freedom defined on $L_2(\mathbb{R}$)). Expanding the initial state $|i\rangle$ in the eigenbasis of an internal \emph{control observable} $\mathbf{A}|a_j\rangle=\lambda_j|a_j\rangle$ one has
\begin{equation}
|i\rangle|\psi\rangle=\sum_jc_j|a_j\rangle|\psi\rangle \,,
\end{equation}
where $c_j$ are an appropriate set of normalised amplitudes and $|\psi\rangle$ is the normalised transverse component of the particle's spatial degree of freedom. We can expand the spatial degree of freedom, for example, in the $x$-basis $|\psi\rangle = \int \psi(x)\textrm{d}x|x\rangle$. Allow the particle to undergo a displacement, transverse to its direction of propagation and dependent on the internal state. The time evolution operator is
\begin{equation}
U=\textrm{e}^{-ig\mathbf{A}\hat{k}_x} \,,
\label{hamil}
\end{equation}
where $\hat{k}_x$ is the operator corresponding to the $x$-component of the particle's momentum, and $g$ is the coupling constant or `interaction parameter'~\cite{HofmannGogginAlmeida2012}. 
The dynamics is unitary since $\mathbf{A}$ is Hermitian and has real eigenvalues. Thus it generates translations in the meter variable conjugate to $\hat{k}_x$. In this case we have 
\begin{equation}
|a_j\rangle \int \psi(x) \textrm{d}x |x\rangle \rightarrow|a_j\rangle \int \psi(x- \lambda_j g)\textrm{d}x |x\rangle \,,
\label{entangled}
\end{equation}
i.e.~displacements in $x$. The analyses in this paper apply equally if $\hat{k}_x$ and $\hat{x}$ are interchanged, so that the interaction induces shifts in momentum space.

Under this dynamics, each of the eigenstates of the control observable becomes correlated to a separate wavefunction, which, up to the proportionality constant $g$, is peaked around the corresponding eigenvalue. When $\lambda_j g$ is much larger than the `width' of $\psi(x)$, this constitutes a strong measurement of the internal state of the particle by its spatial wavefunction. Formally, this occurs when the overlap $\mathcal{O}_{ij}:=\int \psi^*(x- \lambda_i g)\psi(x- \lambda_j g)\textrm{d}x$ between each pair of shifted wavefunctions is vanishingly small~\footnote{Here `measurement' is a unitary (hence reversible) process driven by the Schr{\"o}dinger Equation. The term is often also used to describe the (irreversible) projection onto a particular quantum state, which is associated with wavefunction collapse. We use `detection' for the latter meaning. See e.g.~Ref.~\cite{Neumann1996}. Related semantic issues are discussed in~\cite{Leggett1989,*AharonovVaidman1989}.}. When $\lambda_j g$ is relatively small and the wavefunctions are no longer well resolved, the measurement is said to be \emph{weak}~\cite{MelloJohansen2010}. 

Now allow the particle to undergo another internal-state-dependent shift, this time in the $y$-direction and at full strength, i.e.~with an interaction Hamiltonian $\propto\mathbf{B}\hat{k}_y$. The product of coupling constant and interaction time should be large enough such that the eigenstates $|f\rangle$ of observable $\mathbf{B}$ are well separated into distinct `branches'.

By selecting only one branch for investigation, one \emph{postselects} the particle into a single final internal eigenstate of $\mathbf{B}$.  Any such state can be expanded in the basis of $\mathbf{A}$, $ |f\rangle=\sum_j c'_j |a_j\rangle$.  The final state in the `success' branch will then be given by 
\begin{equation}
|\psi_{\textrm{wva}}\rangle=\frac{1}{\sqrt{q}}\sum_jc_j \bar{c}'_j \int \psi(x-\lambda_jg)\textrm{d}x |x\rangle\,,
\label{exactwf}
\end{equation}
where the normalization $q$ represents the probability of successful postselection, and depends on $g$ through $\mathcal{O}_{ij}$. 

Aharonov, Albert and Vaidman showed that, when $g$ is small, the effect on the wavefunction is captured by a single quantity $A_w=\langle f|\mathbf{A}|i\rangle/\langle f | i \rangle$~\cite{AharonovAlbertVaidman1988}. This is their well known `weak value', and is generally a complex number~\cite{Jozsa2007}. They arrive at this conclusion as follows:
\begin{align}
|\psi_w\rangle&=&\langle f| \textrm{e}^{-ig\mathbf{A} \hat{k}_x}|i\rangle|\psi\rangle&\\
&=& \left(\langle f|i\rangle - ig\langle f |\mathbf{A}|i\rangle \hat{k}_x+\ldots\right)|\psi\rangle&\\
&\approx&\langle f|i\rangle\left(1 - igA_w \hat{k}_x\right)|\psi\rangle& \,,
\end{align}
so that to linear order in $g$, the distribution over $x$ becomes
\begin{align}
P_{\textrm{wva}}(x)&:=|\langle x|\psi_{\textrm{wva}}\rangle|^2\nonumber\\
&\approx|\psi(x-g\textrm{Re}(A_w))|^2 \,,
\label{shiftx}
\end{align}
where we assume only that the initial spatial wavefunction is real-valued but do not make assumptions about its shape~\cite{SusaShikanoHosoya2012,*Lorenzo2013a,*SusaShikanoHosoya2013a}. If we take the initial momentum space wavefunction to be a Gaussian $\tilde{\psi}(k_x)\propto\textrm{e}^{-k_x^2/4\Delta_{k_x}^2}$, the distribution over $k_x$ transforms to~\cite{KofmanAshhabNori2012}
\begin{align}
\tilde{P}_{\textrm{wva}}(k_x)&:=|\langle k_x|\psi_{\textrm{wva}}\rangle|^2\nonumber\\
&\approx|\tilde{\psi}(k_x-2g\Delta_{k_x}^2\textrm{Im}(A_w))|^2.
\label{shiftkx}
\end{align}
Under these assumptions, the shape of the meter wavefunction is not changed, and the perturbation effect can be modeled through the simple shifts above: for a detailed discussion of departures from this behavior, see Ref~\cite{KofmanAshhabNori2012}. We refer to the conditions necessary for these expressions to be accurate as the AAV approximation.
The motivation behind using postselection is that, where the AAV approximation is valid and when $\langle f | i \rangle\rightarrow 0 $,  the average displacement is much larger than is otherwise possible~\footnote{although not arbitrarily larger~\cite{DuckStevensonSudarshan1989}}. The large displacement is often claimed to be an advantage, especially for overcoming sources of technical noise.

A simple thought experiment is seductive: imagine that ordinarily the shift induced by the weak measurement $\lambda_j g$ is smaller than the width of a single pixel in a digital sensor (such as a CCD)~\footnote{or indeed smaller than the width of a photodiode stepped through finite displacements~\cite{RitchieStoryHulet1991}}. An amplified shift, if large enough, could then be distinguished from no shift at all, whereas previously no such distinction was possible~\cite{BrunnerSimon2010}. This thought experiment presupposes an approach to parameter estimation where only the mean of the probability distribution is important. However, the mean is not always a~\emph{sufficient} statistic. As we shall see, if one makes use of the entire statistical distribution to inform one's estimate of $g$, then a larger displacement is not \emph{per se} advantageous.

\section{Metric}
\label{metric}
An efficient approach to parameter estimation is \emph{maximum likelihood}~\cite{Paris2009}, which has previously been studied in connection with weak values in Refs.~\cite{Hofmann2011,KneeBriggsBenjamin2013,VizaMartinez-RinconHowland2013}. It involves inverting the statistics of a measurement record to find an estimate for a parameter of interest. No other estimation technique can give higher confidence asymptotically (when the number of trials is large)~\cite{Fisher1925}. The ultimate precision of this estimation scheme, and hence of any estimation scheme, is captured by the Fisher information: a powerful tool for the appraisal of parameter estimation protocols through their resultant parametric probability distributions. The set of probability distributions conditioned by the unknown parameter define a \emph{statistical model}. The Fisher information is a functional on such conditional probability distributions, and is defined:
\begin{equation}
F_g[P(s|g)]=\int_{-\infty}^\infty \frac{\left(\partial_g P(s|g)\right)^2}{P(s|g)} \textrm{d}s\, .
\end{equation}
It has the properties of an \emph{information}: it is positive and increases as $P(s|g)$ (the probability of measuring $s$ given a value of $g$) becomes less smooth. In this sense it is an inverse measure of entropy. Here $s$, representing an element of the sample space of the statistical model, can be thought of either as a position or momentum variable -- we will consider both $s=x$ and $s=k_x$ in the following. The derivative $\partial_g\equiv\partial/\partial g$ is taken with respect to the parameter of interest -- in this context we are interested in the interaction parameter $g$, but for other parameters we write, e.g.,~$F_s[\bullet]$. 

The Cram\'{e}r-Rao bound limits, from below, the variance in an estimate $\tilde{g}$ of an unknown parameter $g$ given the observed statistics~\cite{Cramer1946}
\begin{equation}
\textrm{Var}( \tilde{g})\geq \frac{1}{NF_g}\, .
\end{equation}
The bound is given by the inverse of the product of the number of trials $N$ with the Fisher information, making the denominator a good measure of confidence in the unknown parameter -- higher is better. The question about WVA inspired metrology then reduces to this: is there more information about $g$ contained in the postselected distribution over $s$ than when no postselection is performed? 

We define a \emph{standard} strategy as one making no use of a final measurement, and hence no use of postselection. In fact, our benchmark measurement strategy completely ignores the `system' degree of freedom. Beginning from Eq.~(\ref{entangled}), one traces over the internal degree of freedom and measures the particle in the $x$ basis to give
\begin{align}
P_{\textrm{std}}(x)=&\sum_j |c_j|^2|\psi(x-\lambda_jg)|^2 \,.
\end{align}
The weighted sum is a convex combination: by controlling the $c_j$ one can mix the probability densities corresponding to the different eigenvalues of $\mathbf{A}$. An immediate optimization for the standard strategy presents itself: The Fisher information metric is convex~\cite{Cohen1968}, meaning that any such mixing of probability distributions will be sub-optimal. One should therefore choose the $c_j$ to filter the probability distribution with eigenvalue $\lambda_*$ that gives the highest Fisher information. Now $F_g[P{_\textrm{std}}]=F_g[|\psi(x-\lambda_*g)|^2]$. 

The key figure of merit that we shall be concerned with in this paper, is the Cram\'{e}r-Rao bound of the WVA strategy, and how it compares to that of the standard strategy. In the limit of $N\rightarrow \infty$, their ratio is equal to $qF_g[P_\textrm{wva}]$ to $F_g[P{_\textrm{std}}]$. The former is corrected by the probability of successful postselection. This is in accordance with the additivity of the Fisher information over $N$ independent events (of which only $qN$ are available to the WVA strategy). While this quantity may not do justice to all notions of metrological performance, it provides a fundamental bound with a clear operational meaning, which holds when the number of trials $N$ is sufficiently large.
\section{Ideal Detector}
\label{ideal detector}
Although the aim of this paper is to study certain detector imperfections, to begin we consider a stable detector having infinite resolution.

It is instructive to investigate the Fisher information of a general distribution $P(s')=|\psi(s')|^2$ whose argument is shifted $s'= s-\nu g$. Consider the derivative w.r.t.~$g$, which by the chain rule
\begin{align}
\partial_g P(s-\nu g )&= \frac{\partial s'}{\partial g}\frac{\partial P(s' )}{\partial s'} \nonumber \\ 
&=- \nu \partial_{s'}P(s') \,,
\label{chainrule}
\end{align}
and so 
\begin{align}
F_g[P(s-\nu g)]&=\int \nu^2 \frac{(\partial_{s'}P(s'))^2}{P(s')}\textrm{d}s \nonumber \\
&=\nu^2 F_s[P(s)] \,,
\label{shifter}
\end{align}
as $\textrm{d}s=\textrm{d}s'$ under the change of integration variable.
Notably, the information is \emph{independent of the size of the parameter} $g$~\cite{Frieden2004}. The multiplier $\nu$, however, acts as a sort of velocity for the statistical model: it represents the rate at which the probability density changes as $g$ is swept through parameter space. The independence of $F$ on $g$ is an important difference to other measures of metrological performance: for example the signal-to-noise ratio given in~\cite{BarnettFabreMaitre2003} and employed by Refs.~\cite{StarlingDixonJordan2009,DixonStarlingJordan2009,Kedem2012} shows a dependence on the size of the shift itself rather than its rate of change.

Equation (\ref{shifter}) implies that in the standard strategy $\nu$ should be given by the largest eigenvalue $\lambda_*$. The standard strategy, then, prepares the system in the eigenstate of $\mathbf{A}$ having the largest eigenvalue, and performs \emph{no} final measurement on the system. 

In the WVA strategy, a final measurement gives rise to the AAV effect, and by Eq.~(\ref{shiftx}) or (\ref{shiftkx}) one of the branches experiences a statistical velocity given by an anomalously large effective eigenvalue $\textrm{Re}(A_w)$ or $2\Delta_{k_x}\textrm{Im}(A_w)$, respectively. By assuming that the experimenter chooses either to measure in position space or momentum space, and therefore will choose either $A_w=\textrm{Re}(A_w)$ or $A_w=\textrm{Im}(A_w)$, respectively, one can replace the real or imaginary part of the weak value with its modulus. Then taking a ratio of the WVA strategy to the standard strategy gives 
\begin{equation}
\frac{F_g[P_{\textrm{wva}}]}{F_g[P{_\textrm{std}}]}=\frac{ |A_w|^2}{ \lambda_*^2}
\label{realfisher}
\end{equation}
for real weak values or 
\begin{equation}
\frac{F_g[\tilde{P}_{\textrm{wva}}]}{F_g[P{_\textrm{std}}]}= \frac{|A_w|^2 }{ \lambda_*^2}\frac{ 4\Delta_{k_x}^4F_{k_x}[\tilde{P}(k_x)]}{ F_x[P(x)]}
\label{imagfisher}
\end{equation}
for imaginary weak values. Recall that $A_w$ can be much larger than $\lambda_*$. It appears then, that in the postselected runs, not only is the average displacement of the meter amplified but so too the measure of sensitivity as captured by the Fisher information. When pre- and postselection are tuned close to orthogonal this triggers destructive interference in the meter, causing cancellation where the superposed meter states overlap most. This leaves only a wavefunction where the meter states overlap least, i.e. where the meter states are most distinguishable. It is therefore not surprising that the postselected wavefunction can give better performance, because the Fisher information is a measure of distinguishability of neighboring meter states as $g$ is varied. 

The $ 4\Delta_{k_x}^4F_{k_x}[\tilde{P}(k_x)]/ F_x[P(x)]$ term in Eq.~(\ref{imagfisher}) must be carefully interpreted. Note that we continue to use the standard strategy information as a benchmark, which necessarily involves detection in the $x$-basis. If we take $P(x)$ and $\tilde{P}(k_x)$ to be generated from Gaussian wavefunctions, we obtain by Heisenberg's relation~\cite{Frieden1992}	
\begin{equation}
\frac{4\Delta_{k_x}^4F_{k_x}[\tilde{P}(k_x)]}{ F_x[P(x)]}=\frac{4\Delta_{k_x}^2}{F_x[P(x)]}=1.
\label{HUP}
\end{equation}
Employing an imaginary weak value then offers no advantage over a real one.

Of course the benefit due to `amplification' is not without cost: we must correct the amplification factor for the reduced probability of data being selected $|\langle f|i\rangle|^2$ -- a necessary drawback of WVA. Then one finds 
\begin{align}
\frac{|A_w|^2}{\lambda_*^2}|\langle f|i\rangle|^2=\frac{|\langle f |\mathbf{A}|i\rangle|^2}{\lambda_*^2} \leq 1.
\label{ineq}
\end{align}
The inequality follows by expanding $|f\rangle$ and $|i\rangle$ in the eigenbasis of $\mathbf{A}$ and applying the Cauchy-Schwarz inequality:  
\begin{align}
|\langle f |\mathbf{A}| i \rangle |^2 = \left\vert \sum_k \bar{c}'_k \lambda_k  c_k \right\vert^2 &\leq \sum_n |c'_n |^2 \sum_m |  \lambda_m c_m |^2 \nonumber \\ 
 &= \sum_m  \lambda_m^2 | c_m |^2\nonumber\\
 &=\langle \mathbf{A}^2 \rangle \leq \lambda_*^2 \, .
 \label{eq:inequality}
\end{align}
Thus, combining Eqs.~(\ref{realfisher}), or (\ref{imagfisher}) and (\ref{HUP}), with (\ref{ineq}) we have
\begin{equation}
\frac{qF_g[P_{\textrm{wva}}]}{F_g[P{_\textrm{std}}]}\leq1;
\end{equation}
at least under the AAV approximation and in the absence of noise, there can be no advantage through WVA for the purpose of estimating $g$. 

The above conclusions are consonant with a very recent result of Tanaka and Yamamoto~\cite{TanakaYamamoto2013}. In their work, the authors calculate the quantum Fisher information~\footnote{ which is a maximal Fisher information varying over all measurement POVMs} of the weak-value amplified wavefunction (suitably postselected), showing that it is never greater than in \emph{the joint system-meter state after the weak measurement}. By contrast, our approach compares against a benchmark that ignores, or `traces-out' the system state during the detection.
\section{Failed postselection}
\label{failedpostselection}

One might wonder whether the large number of discarded events occurring at the postselection stage may be of use. One recent proposal suggested `recycling' the rejected particles for another run~\cite{DresselLyonsJordan2013}. By contrast, in a previous paper we considered looking at the statistics in both `pass' and `fail' branches of the postselection. Calculating the sum of Fisher informations in both branches, we found that (although the Fisher information is by definition greater) the ratio of informations remains at unity or below~\cite{KneeBriggsBenjamin2013}. Here we consider the use of data in both branches as well as their correlations with a `which branch variable' (call this $b$). In other words, we upgrade the WVA strategy; replacing postselection with a final measurement and allowing any postprocessing of the result. The ultimate performance of this class of estimation strategies will be given by the Fisher information calculated in the \emph{joint} probability distribution over $b,s$:
\begin{equation}
F_g[\mathcal{P}(b,s)]:=\sum_b\int \textrm{d}x \frac{(\partial_g\mathcal{P}(b,s))^2}{\mathcal{P}(b,s)}.
\end{equation}
A joint distribution generally contains more structure (and more information) than is revealed by its marginal distributions. Begin with
\begin{equation}
\mathcal{P}(b,x)=\sum_f \delta_{fb} q_b |\psi(x-\textrm{Re}(A_w^b)g)|^2,
\end{equation}
or
\begin{equation}
\tilde{\mathcal{P}}(b,k_x)=\sum_f \delta_{fb} q_b |\tilde{\psi}(k_x-2\Delta_{k_x}\textrm{Im}(A_w^b)g)|^2,
\end{equation}
once more employing the AAV approximation: $b$ indexes the outcomes of the final, strong measurement. In this regime, $q$ does not depend on $g$, although in an real experiment the postselection probability itself carries information about the interaction parameter~\cite{ZhangDattaWalmsley2013}. Because of the Kronecker $\delta$, each output branch can be considered independently, leading to 
\begin{align}
F_g[\mathcal{P}(b,x)]&=\sum_f q_f |A_w^f|^2 F_x[P(x)]\nonumber \\
&= \langle \mathbf{A}^2\rangle F_x[P(x)]
\end{align}
for real weak values, or indeed
\begin{align}
F_g[\tilde{\mathcal{P}}(b,k_x)]&=\sum_f q_f |A_w^f|^2 (2\Delta_{k_x})^2 \nonumber \\
&= \langle \mathbf{A}^2\rangle4\Delta_{k_x}^2
\end{align}
for imaginary weak values. The last step is in strong analogy with the usual resolution of probability-weighted weak-values to the expectation value $\sum_f q_f A_w^f=\langle \mathbf{A} \rangle $~\cite{ShikanoHosoya2010}. Notice that once the initial state of system and meter are chosen, the expected value of the square of the control observable along with the shape of initial meter wavefunction set a fixed amount of information in the joint system-meter state after their interaction. The final strong measurement may distribute the information (which we should think of as a conserved quantity) among the output branches in various ways. The choice of basis for the final measurement may lead to a balanced distribution if, for example, the final measurement basis is unbiased with respect to the initial state. If on the other hand the basis contains an element which is almost orthogonal to the initial state, an arbitrarily large portion of the information may be concentrated into the corresponding branch of the final measurement. We repeat this argument in Section~\ref{exactcalculations_failedpostselection} without making approximations.

When the AAV effect gives rise to a significant amplification, there is little to be gained from monitoring other branches: henceforth we shall thus only be concerned with the probability distribution $P_{\textrm{wva}}$ in a suitably chosen `success' branch. In such a situation an important question remains: is the postselected meter wavefunction (carrying almost all of the information about $g$) more robust to technical noise than an un-postselected wavefunction? 
\section{Pixelation}
\label{pixelation}
We are now in a position to introduce realistic imperfections to the scenario. We consider first the effects of pixelation on an arbitrary wavefunction which has been subject to a simple shift. This is of relevance to real weak values, and also to imaginary weak values for a restricted class of momentum space wavefunctions. We then specialize to Gaussian wavefunctions, which serves both as an instructive example and also as the canonical wavefunction used in weak-value amplification experiments for both real and imaginary weak values. 

\subsection{For any wavefunction}
Imagine that the particle (after weak measurement and postselection) impinges on a detector comprising pixels of finite width. To model such a device, we build up a discrete probability mass $\textrm{Pr}(n)$ by dividing the $s$-axis into pixels of size $r_s$. Each pixel carries an integer label $n=\lfloor s/r_s\rceil$ (so that values of $s$ are rounded to the nearest integer multiple of $r_s$). The total Fisher information then becomes
\begin{equation}
F_g[\textrm{Pr}(n)]=\sum_n\frac{1}{\textrm{Pr}(n)}(\partial_g \textrm{Pr}(n))^2 \,.
\end{equation}
The width of the detector is modeled by varying the range of this sum, but here we take it to be infinite.  For this reason any relabeling of the pixels will not change the result. For example, adding a fixed integer to each pixel's label is a bijection $f: \mathbb{Z}\rightarrow\mathbb{Z}$, which preserves the value of the sum in analogy with the above change of integration variable, equation~(\ref{shifter}). The probability of a click in pixel $n$ is
\begin{equation}
\textrm{Pr}(n)= \int_{r_s(n-1/2)}^{r_s(n+1/2)} P(s)\textrm{d}s \,.
\label{pixelate}
\end{equation}
Because of the invariance under a relabeling of the pixels, any shifts in $s$ can now be taken modulo $r_s$. 

To understand the effect of pixelation on $P(s-\nu g)$ we may use the chain rule (\ref{chainrule}) again, taking the derivative under the integral sign:
\begin{align}
F_g[\textrm{Pr}(\lfloor s'/r_s\rceil)]&=\sum_n\frac{\left(\int_{r_s(n-1/2)}^{r_s(n+1/2)}\partial_g P(s') \textrm{d}s\right)^2}{\int_{r_s(n-1/2)}^{r_s(n+1/2)}P(s') \textrm{d}s}\nonumber \\
&=\nu^2\sum_n\frac{\left(\partial_{s'}\int_{r_s(n-1/2)}^{r_s(n+1/2)} P(s') \textrm{d}s\right)^2}{\int_{r_s(n-1/2)}^{r_s(n+1/2)}P(s') \textrm{d}s} \,.
\end{align}
The pixelation effect is almost decoupled from the dependence on $\nu$. However, when making the change of integration variable,
\begin{align}
\textrm{Pr}(\lfloor s'/r_s\rceil)=&\int_{r_s(n-1/2)}^{r_s(n+1/2)}P(s-\nu g)\textrm{d}s \nonumber \\
=&\int_{r_s(n-1/2-\nu g)}^{r_s(n+1/2-\nu g)}P(s')\textrm{d}s' \,,
\end{align}
the finite limits of integration prevent removing the dependence on $\nu$ altogether. Allowing the pixelated detector to be displaced through a controllable quantity $\mu$: 
\begin{align}
\textrm{Pr}(\lfloor s'/r_s\rceil)&=\int_{r_s(n-1/2)}^{r_s(n+1/2)}P(s+\mu-\nu g)\textrm{d}s\nonumber \\
&=\int_{r_s(n-1/2)}^{r_s(n+1/2)}P(s-h)\textrm{d}s \,,
\end{align}
it becomes clear that the importance of $\nu$ is screened by the alignment $h:=(\nu g-\mu)~\textrm{mod~}r_s$. The fraction $h$ indicates the relative alignment of the centroid with the pixel boundaries. For example, if $h=0$ then the mean is aligned exactly in the middle of a pixel; if $h=0.5$ then the mean is on a pixel boundary -- this is shown schematically in the inset of Fig.~\ref{Gaussian_resolution}. One then obtains
\begin{align}
F_g[\textrm{Pr}(\lfloor (s-\nu g) /r_s\rceil)]=\nu^2 F_s[\textrm{Pr}(\lfloor (s-h)/r_s\rceil)].
\end{align}
The relationship of $\nu$ to $h$ is purely incidental: there is no reason to suppose that larger values of $\nu$ will make alignment easier than smaller values. One should therefore treat $h$ identically for the WVA and standard strategies. One can consider adaptive techniques to asymptotically achieve $h = 0.5 $, or -- if alignment control is impossible -- one can average over $h\in[0,0.5]$. This is the first key result of our paper: \emph{Any degradation due to pixelation commutes with the amplification effect.}

Correcting by the success probability and dividing by the standard strategy information, for real weak values one has, taking $h$ identical in both strategies, 
\begin{align}
\frac{qF_g[\textrm{Pr}(\lfloor (x-\textrm{Re}(A_w)g)/r_x\rceil)]}{F_g[\textrm{Pr}(\lfloor (x-\lambda_*g)/r_x\rceil)]}=&\frac{\textrm{Re}(\langle f |\mathbf{A}|i\rangle)^2}{\lambda_*^2}\leq&1 \,.
\end{align}
The pixelation effect has completely cancelled, reducing the problem to that of the ideal detector above -- see~(\ref{ineq}).

For imaginary weak-values, we have
\begin{align}
\frac{qF_g[\textrm{Pr}(\lfloor (k_x-4\Delta_{k_x}\textrm{Im}(A_w)g)/r_{k_x}\rceil)]}{F_g[\textrm{Pr}(\lfloor (x-\lambda_*g)/r_x\rceil)]}\leq\frac{\alpha(r_{k_x},\Delta_{k_x})}{\alpha(r_x,\Delta_x)} \,,
\end{align}
where we define $\alpha(r_s,\Delta_s)=F_s[\textrm{Pr}(\lfloor s \rceil)]/F_s(P(s))$ as the fraction of  information remaining after pixelation, and have applied Eqs.~(\ref{HUP}) and ~(\ref{ineq}). Clearly there will not be a perfect cancellation if the pixelation in the $k_x$ direction is more or less severe than in the $x$ direction: but as we have shown this is independent of the magnitude of the mean. Moreover we argue in the next section that $\alpha$ is monotonically decreasing in $r_s/\Delta_s$, and in fact equal to unity to first order. Then $\alpha(r_{k_x},\Delta_{k_x}):\alpha(r_x,\Delta_x) \approx 1$. 
\subsection{For Gaussian wavefunctions}
It is instructive to fix the input wavefunction to get an idea of exactly how the pixelation affects its informational content.  By letting the initial state be described by 
\begin{equation}
\psi(s)=\sqrt{\frac{1}{\sqrt{2\pi} \Delta_s}}\exp\left(-\frac{s^2}{4\Delta_s^2}\right) \,,
\label{gauss}
\end{equation}
i.e.~a Gaussian centered on the origin with waist $\Delta_s$, one finds 
\begin{equation}
F_s[P(s)]=\frac{1}{\Delta_s^2} \,.
\label{advantage}
\end{equation} 
so that $F_g[P(s')]=\nu^2/\Delta_s^2$.
Under pixelation $P(s')$ becomes a discrete probability mass. By the above arguments it is sufficient to study
\begin{align}
\textrm{Pr}(n) =& \frac{1}{2}\left[\textrm{erf}\left(\frac{r_s}{\Delta_s}\gamma_+\right)-\textrm{erf}\left(\frac{r_s}{\Delta_s}\gamma_-\right)\right] \,,
\end{align} 
where we have introduced convenient quantities $\gamma_\pm=(h-n\pm 0.5)/\sqrt{2}$ and applied Eq.~(\ref{pixelate}).
Define $\alpha$ by dividing the diminished Fisher information of the pixelated probability mass by the un-pixelated information
\begin{align}
\alpha(r_s,\Delta_{r_s})=\frac{F_s[\textrm{Pr}(n)]}{F_s[P(s)]}=&\frac{1}{F_s[P(s)]}\sum_n f_n\, ,
\end{align}
where the summand is to be understood as the pixel-wise information. This function shows how much information remains after pixelation has occurred. We may take the division inside the sum, attaching it to each term. In pixel $n$ 
\begin{equation}
\frac{f_n}{F_s[P(s)]}=\frac{1}{\pi} \frac{ \left(e^{- R^2\gamma_+^2}-e^{- R^2\gamma_-^2} \right)^2}{\text{erf}(R\gamma_+)-\text{erf}(R\gamma_-)}\,,
\label{relativeinformations}
\end{equation}
where we have chosen to define $R:=r_s/\Delta_s$ as an inverse measure of resolution. It represents the granularity of the pixels against the characteristic width of the Gaussian. One expects the ratio of $F_s[\textrm{Pr}(n)]$ to $F_s[P(s)]$ to converge to unity as $R\rightarrow0$. 

All that remains is to sum these relative informations \eqref{relativeinformations} over all pixels.  We numerically plot the dependence of the sum on $R$ in Fig.~\ref{Gaussian_resolution}. The limit when $h=0.5$ and $r_s\rightarrow\infty$ is analytically known: the Fisher information is a fraction $2/\pi$ of the non-pixelated case~\cite{PotzelbergerFelsenstein1993}. This limit corresponds to a perfectly aligned  `split-detector' with only two pixels~\cite{StarlingDixonJordan2009}. Note that in the opposite limit, $\alpha\approx 1$ to first order in $R$, meaning that pixelation is a non-issue as long as the pixel size is at least as small as the width of the wavefunction. Our results agree qualitatively with Ref.~\cite{StrubiBruder2013}, which found that pixelation does not set a lower bound on the magnitude of a shift that can be estimated.

\begin{figure}[t]
\includegraphics[width=\linewidth]{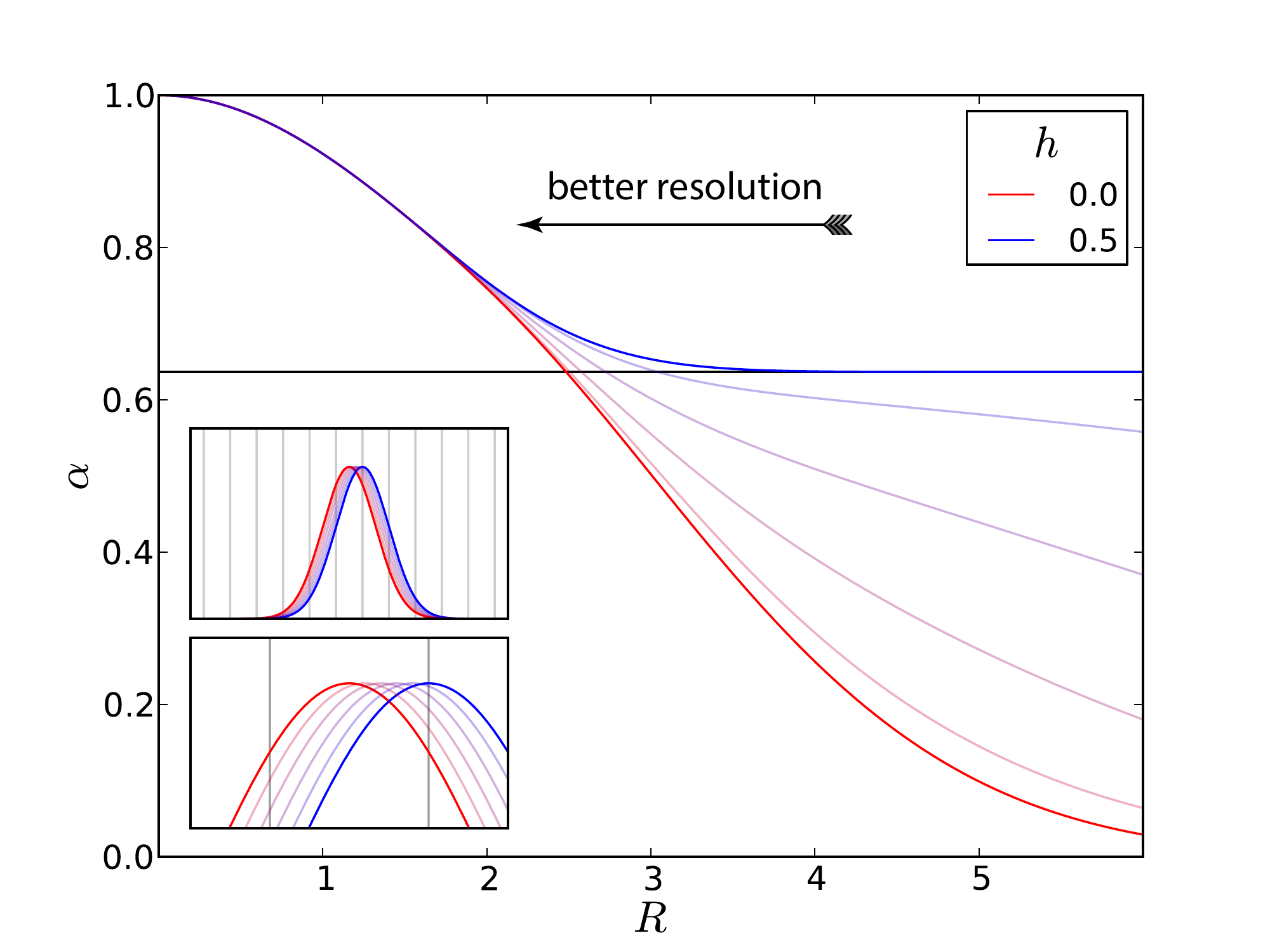}
\caption{\label{Gaussian_resolution}Numerically obtained relationships between relative Fisher information $\alpha$ and inverse resolution $R$. The different curves (blue to red, top to bottom and right to left in the inset) correspond to different alignments which are schematized in the insets. The worst and best cases at $h=0$ and $h=0.5$, respectively, are shown in bold. Aside from misalignment effects (which become important when $R\gtrsim3$) finite resolution does not dramatically limit parameter estimation. In fact with only two pixels the penalty paid is only a loss of about a third of the information, as long as good alignment is possible. The curves are equally relevant for the standard strategy or the WVA strategy.}
\end{figure}
\begin{figure}[t] 
\includegraphics[width=\linewidth]{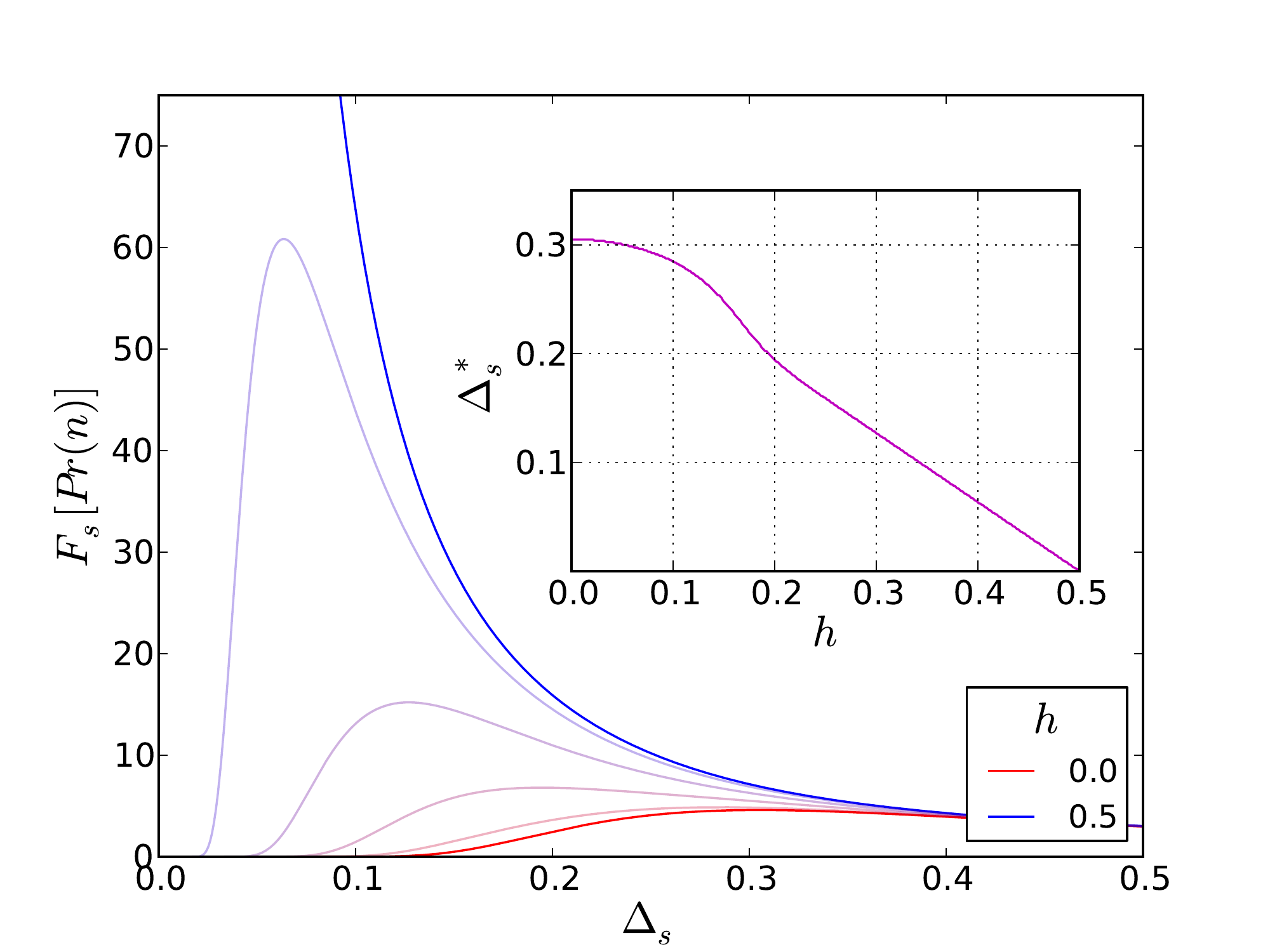}
\caption{\label{uncertainty_vs_alignment}Numerically obtained relationships between Fisher information and Gaussian spread $\Delta_s$ for fixed pixel width $r_s=1$. The bold curves are for $h=0$ (lowermost) and $h=0.5$ (uppermost), with fainter curves interpolating linearly. Without pixelation, narrower Gaussian distributions are superior -- the delta function offers the ultimate precision. However, the misalignment effect introduced with pixelation combines with the low effective resolution as $\Delta_s$ is reduced, preventing the continued improvement. For $\Delta_s \gtrsim r_s /3$ alignment ceases to be important. At any given degree of misalignment (which is described by $h$), there will be an optimum choice $\Delta_s=\Delta_s^*$. The inset shows how this optimum changes with $h$. }
\end{figure}

When pixelation is severe and alignment is poor one does better with a broader statistical distribution than with a narrower one. Taking the wavefunction once more to be Gaussian \eqref{gauss} in Fig.~\ref{uncertainty_vs_alignment} we fix the pixel width and show how the Fisher information depends on $\Delta_s$ in such a way as to offer a trade-off between uncertainty and misalignment. If perfect alignment cannot be achieved, there is a limit to how low $\Delta_s$ may be taken before the worsened effective resolution combines with the bad alignment to kill off the Fisher information. For poor enough resolution, it can thus be advantageous to increase $\Delta_s$ -- introducing extra uncertainty leads to a superior performance. One could increase $\Delta_s$ through a unitary process, to present a broader meter wavefunction for the system to interact with. In circumstances of extreme pixelation, this may be thought of as making the measurement weaker to gain an advantage -- but this is quite distinct from the weak-values formalism (there is no postselection), and should not be associated with AAV's technique. Another option would be to intentionally introduce random displacements to the meter (or detection apparatus) immediately prior to detection. The latter option is an instance of a well-known image- and audio-processing technique known as dithering, and does not affect the strength of the measurement.

\section{Jitter}
\label{jitter}
Here we model another prevalent source of imperfection: random lateral displacements of the particle. Again we begin with an arbitrary wavefunction subject to a simple shift (composed of a $g$ dependent part and also a random part), but then specialize to Gaussian wavefunctions in order to make a concrete connection with the majority of WVA experiments. 
\subsection{For any wavefunction}
We define `jitter' -- random displacements of the measuring apparatus, or equivalently the incident beam -- by convoluting the probability distribution with a suitable noise kernel $\mathcal{N}_s$:
\begin{equation}
\mathcal{N}_s\star P(s):=\int P(s+\mu)\mathcal{N}_s(\mu) \textrm{d}\mu \,.
\end{equation}
This is perhaps the most prominent source of technical noise that has been studied in the WVA context, until now only using the signal-to-noise ratio metric~\cite{BarnettFabreMaitre2003,StarlingDixonJordan2009,Kedem2012,KofmanAshhabNori2012}. In those cases WVA gave a relative advantage, due to identifying the `signal' with the mean of the probability distribution, which becomes invisible when it falls through the noise floor. The Fisher information metric informs us about the performance of a superior estimation strategy, which allows for (in principle) \emph{any} data to contribute to the estimate of the unknown parameter. In our application of the metric, we assume the noise is independent of $g$ so that, by another application of the chain rule (\ref{chainrule}), we have
\begin{align}
F_g[\mathcal{N}_s\star P(s-\nu g)]&=&\int \frac{(\partial_g [ \mathcal{N}_s \star P(s')])^2}{\mathcal{N}_s\star P(s')}\textrm{d}s\nonumber\\
&=&\nu^2\int \frac{(\partial_{s'} [\mathcal{N}_s\star P(s')])^2}{\mathcal{N}_s\star P(s')}\textrm{d}s\nonumber\\
&=&\nu^2F_x[\mathcal{N}_s\star P(s)] \,.
\end{align}
Once more \emph{the application of noise commutes with the amplification factor $\nu^2$}, the second main result of this paper. For real weak values, the noise cancels out in the ratio of Fisher informations:
\begin{equation}
\frac{qF_g[\mathcal{N}_x\star P(s-\textrm{Re}(A_w)g)]}{F_g[\mathcal{N}_x\star P(s-\lambda_*g)]} = \frac{\textrm{Re}(\langle f |\mathbf{A}|i\rangle)^2}{\lambda_*^2}\leq 1\,.
\end{equation}
For imaginary weak values, once more there may not be a complete cancellation due to the possibility of different width wavefunctions and different noise kernels in position and momentum space. We then have
\begin{align}
\frac{qF_g[\mathcal{N}_{k_x}\star P(k_x-4\Delta_{k_x}\textrm{Im}(A_w)g)]}{F_g[\mathcal{N}_x\star P(x-\lambda_*g)]} \leq \frac{\beta(\Delta_{k_x},\mathcal{N}_{k_x})}{\beta(\Delta_{x},\mathcal{N}_{x})}
\end{align}
on defining $\beta(\Delta_s,\mathcal{N}_s)$ as the attenuation factor of Fisher information under jitter: a function which depends on the width of the wavefunction and on the noise kernel in $s$-space (see below). We note that $\beta(\Delta_{k_x},\mathcal{N}_{k_x})$ may or may not be independent of $\beta(\Delta_x,\mathcal{N}_x)$, depending, e.g., on how detection is achieved in the laboratory. Exploiting the relative severity of noise in position and momentum space with imaginary weak values has been proposed by Kedem~\cite{Kedem2012}; but this possibility does not derive from the amplification of the mean.

Note that while we require $\partial_g \mathcal{N}_s=\partial_s \mathcal{N}_s=0$, otherwise the noise kernel $\mathcal{N}_s(\mu)$ can have arbitrary properties (up to normalization and positivity); in particular it may have non-zero mean. 

\subsection{For Gaussian wavefunctions}
As an example, choose the noise kernel to be normally distributed (with zero mean and variance $J_s$), giving
\begin{equation}
\mathcal{N}\star P(s)= \int_{-\infty}^\infty \frac{P(s-\mu)\exp\left(-\frac{1}{2}\frac{\mu^2}{J_s^2}\right)}{\sqrt{2\pi}J_s}\textrm{d}\mu \,.
\end{equation}
If $P(s;\Delta_s)$ is Gaussian with width $\Delta_s$, the application of jitter is equivalent to a redefinition of the spread
\begin{equation}
\mathcal{N}\star P(s-\nu g;\Delta_s)=P(s-\nu g;\sqrt{\Delta_s^2+J_s^2})
\end{equation}
so that the Fisher information takes on a Lorentzian dependence on $\Delta_s$, with scale parameter proportional to $J_s$
\begin{equation}
F_g[\mathcal{N}_s\star P(s-\nu g;\Delta_s)]=\frac{\nu^2}{\Delta_s^2+J_s^2} \,.
\end{equation}
We then find, for this example,
\begin{equation}
\beta(\Delta_s,J_s)=\frac{1}{1+\frac{J_s^2}{\Delta_s^2}}.
\end{equation}
\section{exact calculations}
\label{exactcalculations}
Many experiments have been suggested and performed~\cite{RitchieStoryHulet1991,HostenKwiat2008,BrunnerSimon2010,VizaMartinez-RinconHowland2013,StrubiBruder2013,HofmannGogginAlmeida2012,StarlingDixonJordan2009,XuKedemSun2013} in the regime where the AAV approximation holds. Our discussion above immediately implies that an amplified mean can be no more robust against pixelation or jitter than a standard strategy. However, the AAV approximation ignores certain non-linear effects of the system-meter interaction and can overestimate the centroid of the probability distribution~\cite{Di-Lorenzo2012,KofmanAshhabNori2012, Lorenzo2013}. When the approximation breaks down, the position of the centroid can deviate from $A_w$ and the wavefunction shape changes, sometimes featuring a secondary peak~\cite{DuckStevensonSudarshan1989,WuLi2011,Geszti2010}. For simplicity, and because it is almost exclusively the case in experiments, we restrict the dimension of the system to 2.

We first show that our result for noise-free detection -- that there can be no advantage from WVA -- holds independently of the AAV approximation. Additionally, we provide numerical evidence that the main conclusions of this paper -- with respect to jitter and pixelation -- persist under a non-perturbative treatment of the AAV effect. The arguments are arranged in two subsections: first we choose pre- and postselection to guarantee a real weak value, and then we make a choice to guarantee a pure imaginary weak value. 
\subsection{Real weak-value amplification}
\begin{figure*}
\includegraphics[width=\linewidth]{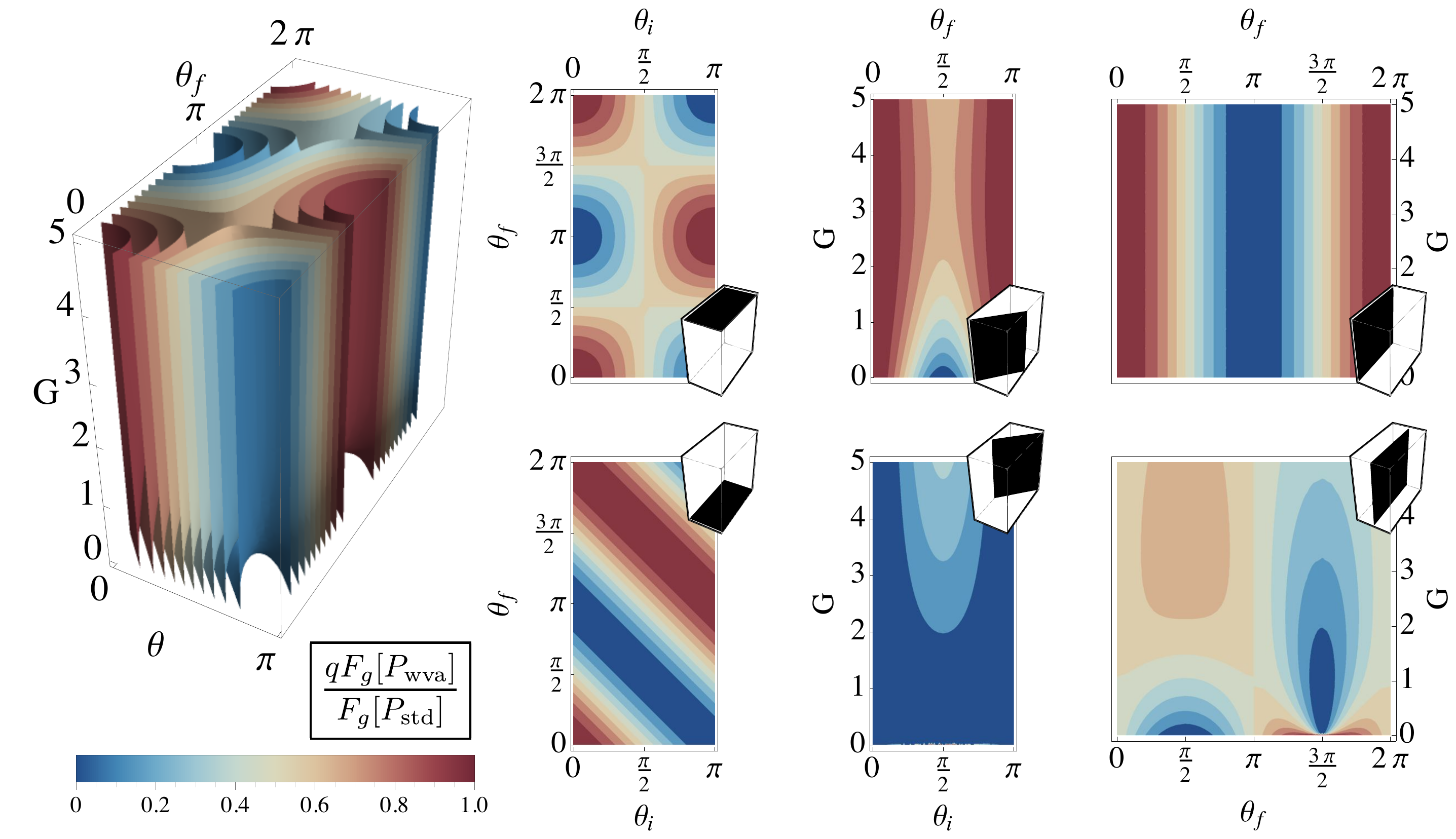}
\caption{\label{full_parameter_space}The corrected Fisher information for the WVA strategy is shown in units of the standard strategy information as a contour plot in the measurement strength $G$ and the pre- and postselection angles $\theta_i,\theta_f$. Two-dimensional slices through this three-dimensional plot are taken for interesting parameter combinations. \textbf{The three uppermost contour plots} indicate parameter combinations for which WVA can mimic the standard technique (red regions). From left to right: In the limit of a strong measurement $G\rightarrow\infty$;  when preselection and postselection are parallel $\theta_i=\theta_f$; and when preselection is into an eigenstate of the control observable $\theta_i=0$. \textbf{The three lowermost contour plots} indicate combinations for which the WVA technique can exhibit characteristically amplified centroids . From left to right: in the limit of a weak measurement $G\rightarrow 0$; when pre- and postselection are chosen orthogonal $\theta_f=\theta_i+\pi/2$; and when preselection is into a state unbiased w.r.t. the control observable $\theta_i=\pi/2$. The three-dimensional plot isosurfaces have colors corresponding to fractional values of $qF_g[P_{\textrm{wva}}]/F_g[P{_\textrm{std}}]$ (the colourmap is given in the legend): in the two-dimensional plots the same colors shade regions between contour lines. The symmetries of the parameter space allow us to restrict $\theta_i\in[0,\pi]$ with no loss in generality.}
\end{figure*}

We make the following parameterization of the internal state of the particle:
\begin{align}
|i\rangle =& \ci |+\rangle + \si |-\rangle \,, \\
|f\rangle =& \cf |+\rangle + \sf |-\rangle \,,
\end{align}
where $\theta_i$ and $\theta_f$ are the pre- and postselection angles, respectively. The relative phases of these states have been set to unity in the eigenbasis of the control observable, $\mathbf{A}|\pm\rangle = \pm |\pm\rangle$. This guarantees a purely real weak value: an optimal choice for detection in the position basis.

With the above parameterization one finds, without approximation and by equation~(\ref{exactwf}), that
\begin{equation}
\psi_{\textrm{wva}}(x)=\frac{1}{\sqrt{q}}\left(\ci\cf\psi_+ + \si\sf\psi_-\right) \,,
\label{exactwfparam}
\end{equation}
where $\psi_\pm=\psi(x\pm g)$.
In the absence of noise, one can find an analytic expression for the ratio of informations available under each strategy, again taking $\psi(x)$ as a Gaussian (\ref{gauss}),
\begin{widetext}
\begin{equation}
\frac{qF_g[P_{\textrm{wva}}]}{F_g[P{_\textrm{std}}]}=\frac{1}{2} \left(1+\cos\theta_f \cos\theta_i -e^{-\frac{G}{2}} \sin\theta_f \sin\theta_i +\frac{G (1+\cos\theta_f \cos\theta_i )}{1+e^{G/2} \textrm{cosec}\,\theta_f \textrm{cosec}\,\theta_i   (1+\cos\theta_f \cos\theta_i )}\right) \,.
\label{ratioexact}
\end{equation}
\end{widetext}
We have defined $G=g^2/\Delta_x^2$ as the \emph{measurement strength}. It is instructive to examine (\ref{ratioexact}) in certain limiting cases, shown in Fig.~\ref{full_parameter_space}. Treating the situation at this level of generality allows for the study of intermediate parameter regimes -- for example when the weak measurement back-action is non negligible~\cite{FeizpourXingSteinberg2011}. The ratio is entirely symmetric with respect to interchanging $\theta_i \leftrightarrow \theta_f$, and is also invariant under $\theta_i+\pi,\theta_f+\pi$. Whilst the ratio (\ref{ratioexact}) can reach unity, it never exceeds it.

To study jitter and pixelation, we performed a numerical search for an advantage under the WVA technique using the exact wavefunction \eqref{exactwfparam}. Firstly, under jitter a numerical maximization of $qF_g[\mathcal{N}_x\star P_{\textrm{wva}}(x)] / F_g[\mathcal{N}_x\star P_{\textrm{std}}(x)]$ over the $\{g,\Delta_x,\theta_i,\theta_f,J_x\}$ parameter space suggests that it remains at unity or below for all combinations. For pixelation, a numerical maximisation tends to find parameter combinations that lead to a ratio greater than one, because of the alignment sometimes being incidentally superior.

\subsection{Imaginary weak-value amplification}
Here we parameterize the qubit by
\begin{align}
|i\rangle =& \frac{1}{\sqrt{2}}( |+\rangle + \textrm{e}^{i\phi_i} |-\rangle) \,, \\
|f\rangle =& \frac{1}{\sqrt{2}}( |+\rangle + \textrm{e}^{i\phi_f}  |-\rangle) \,,
\end{align}
which guarantees a purely imaginary weak value. The position wavefunction is then 
\begin{equation}
\psi_{\textrm{wva}}(x)=\frac{1}{\sqrt{2q}}\left(\psi(x+g)+\textrm{e}^{i(\phi_i-\phi_f)}\psi(x-g)\right)\, ,
\end{equation}
and a Fourier transform yields the momentum wavefunction 
\begin{equation}
\tilde{\psi}_{\textrm{wva}}(k_x)=\frac{1}{\sqrt{2q}}\tilde{\psi}(k_x)\left(\textrm{e}^{-igk_x}+\textrm{e}^{i(\phi_i-\phi_f+gk_x)}\right)\,.
\label{momentumwf}
\end{equation}
We choose the spatial wavefunction to be a Gaussian with waist $\Delta_x$, so that the momentum wavefunction is also Gaussian with waist $\Delta_{k_x}=1/2\Delta_x$. Thus the measurement strength $G=g^2/\Delta_x^2=4g^2\Delta_{k_x}^2$. Computing the corrected ratio of informations leads to 
\begin{equation}
\frac{qF_g[P_{\textrm{wva}}]}{F_g[P{_\textrm{std}}]}=\frac{2 e^{G/2} G \cos (\delta\phi)+2 e^G-\cos (2\delta\phi-1)}{4 e^{G/2} \cos (\delta\phi)+4 e^G}.
\end{equation}
Once more we compare against the standard strategy which measures the particle in the $x$-basis. We have used the shorthand $\delta\phi=\phi_i-\phi_f$: It is clear that only this difference in the relative phase angles matters, and the parameter space consists of only $\{G,\delta\phi\}$ . By inspection the ratio can never exceed unity, see Fig.~\ref{imaginarycontourplot}. 

We performed a numerical search to investigate the effects of jitter on the exact wavefunction, equation~(\ref{momentumwf}). Constraining the search so that the jitter is equally severe in momentum space and position space (this corresponds to $J_x \Delta_{k_x}= J_k \Delta_x=J_{k_x} /2\Delta_{k_x}$), we find the ratio of informations is bounded by unity throughout the $\{g,\Delta_{k_x},\phi_i-\phi_f, J_x, J_{k_x}\}$ space. This confirms the situation that was found in the AAV limit. A similar investigation for pixelation sometimes gives a ratio of higher than unity due to the aforementioned incidental alignment problem.

\begin{figure}
\includegraphics[width=\columnwidth]{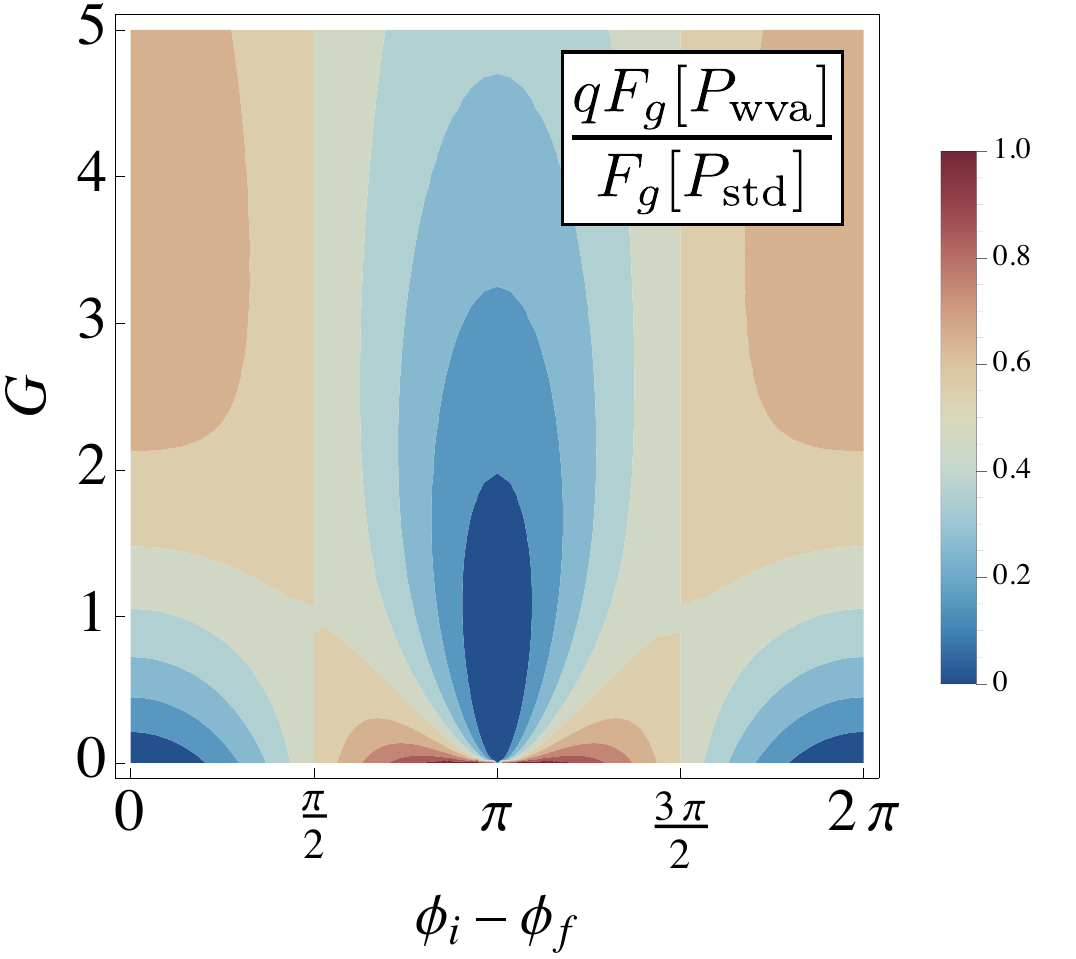}
\caption{\label{imaginarycontourplot}Contour plot of the corrected ratio of Fisher information for the weak-value amplification strategy to Fisher information for the standard strategy, expressed as a function of the measurement strength $G$ and angle between pre- and postselection $\phi_i-\phi_f$. The plot relates to imaginary weak values.}
\end{figure}
\section{exact calculations - failed postselection}
\label{exactcalculations_failedpostselection}
Here we repeat the calculations of Section \ref{failedpostselection} relating to failed postselection without the use of any assumptions: in particular the postselection branch index $b$ is allowed to carry information about $g$, and we include all non-linear effects of the system-meter interaction. We restrict the system to a two-dimensional Hilbert space, so that $b$ is now a binary pass/fail variable. We first present an argument with real weak values, and then an argument for imaginary weak values. In both cases the wavefunction can take an arbitrary shape. In both cases we shall make use of the following fact: if one has a joint distribution
\begin{equation}
\mathcal{P}(b,s)=\delta_{b1}\bar{P}^{\textrm{pass}}+\delta_{b0}\bar{P}^{\textrm{fail}},
\end{equation}
where the bar denotes an unnormalized probability distribution, the Fisher information will be
\begin{align}
F_g\left[\mathcal{P}(b,s)\right]:=&\sum_{b}\int \textrm{d}s \frac{\left(\partial_g \mathcal{P}(b,s)\right)^2}{\mathcal{P}(b,s)}\nonumber\\
=&\int\textrm{d}s\frac{\left(\partial_g \bar{P}^{\textrm{pass}}\right)^2}{\bar{P}^{\textrm{pass}}}\nonumber \\
&+\int \textrm{d}s\frac{\left(\partial_g \bar{P}^{\textrm{fail}}\right)^2}{\bar{P}^{\textrm{fail}}}\,.
\label{jointfisher}
\end{align}
\subsection{Real weak values}
\begin{lemma}
\label{Unnormalised probability density functions}
For $\psi=\sum_i e_i \psi_i$ and $\bar{P}(x)=\psi^2$ a (possibly unnormalized) probability density we have 
\begin{align}
\int \textrm{d}x \left[\left(\partial_g \bar{P}\right)^2/\bar{P}\right]=&\sum_i e_i^2 F_g[\psi_i^2]\nonumber\\
&+4\sum_{i\neq j }e_ie_j\int\textrm{d}x\partial_g\psi_i\partial_g\psi_j\,.
\end{align}
Where $F_g[\bullet]$ is the Fisher information functional. Note that the expression is not an information unless $\bar{P}$ is normalised: otherwise it may be negative.
\end{lemma}

\begin{proof}
Consider
\begin{equation}
\frac{\left(\partial_g \bar{P}\right)^2}{\bar{P}}=\frac{\left(2\psi\partial_g\psi\right)^2}{\psi^2}=4\left(\partial_g\psi\right)^2\,,
\end{equation}
which follows by the chain rule. Now write
\begin{equation}
\left(\partial_g\psi\right)^2=\left(\sum_ie_i\partial_g\psi_i\right)\left(\sum_je_j\partial_g\psi_j\right)
\end{equation}
by taking the derivative inside the sum. Expanding the sum into diagonal and off diagonal terms, multiplying by four and integrating over $x$ gives the result. 
\end{proof}
\begin{theorem}
Let  $b=1$ when postselection passes and $b=0$ when postselection fails.

Consider the joint probability distribution
\begin{align}
\mathcal{P}(b,x) &=  \delta_{b1}\left( \ci\cf\psi_++\si\sf\psi_-\right)^2\nonumber\\
&+ \delta_{b0}\left( -\ci\sf\psi_++\si\cf\psi_-\right)^2\nonumber\\
&= \delta_{b1}\bar{P}^{\textrm{pass}}+\delta_{b0}\bar{P}^{\textrm{fail}}\,,
\end{align}
where $\psi_\pm$ is a real wavefunction with arbitrary shape that has undergone a positive (negative) shift by some function of $g$. It has Fisher information equal to that found in either $P_\pm:=\psi_\pm^2$.
\end{theorem}
\begin{proof}
Using Eq.~(\ref{jointfisher}), and by the Lemma, we have 
\begin{align}
F_g\left[\mathcal{P}(b,x)\right]=&\left(\cos^2\frac{\theta_i}{2}\cos^2\frac{\theta_f}{2}+\sin^2\frac{\theta_i}{2}\sin^2\frac{\theta_f}{2}\right)F_g[\psi_+^2]\nonumber\\
+&\left(\cos^2\frac{\theta_i}{2}\sin^2\frac{\theta_f}{2}+\sin^2\frac{\theta_i}{2}\cos^2\frac{\theta_f}{2}\right)F_g[\psi_-^2]\nonumber\\
+&4\left(2\cos\frac{\theta_i}{2}\cos\frac{\theta_f}{2}\sin\frac{\theta_i}{2}\sin\frac{\theta_f}{2}-\right.\nonumber\\
&\left.2\cos\frac{\theta_i}{2}\sin\frac{\theta_f}{2}\sin\frac{\theta_i}{2}\cos\frac{\theta_f}{2}\right)\nonumber\\
&\times\int \textrm{d}x \partial_g\psi_+\partial_g\psi_- \, .
\end{align}
The final term is zero.  By symmetry we have $F_g[\psi_+^2]=F_g[\psi_-^2]$ giving the desired result:

\begin{equation}
F_g[\mathcal{P}(b,x)]=F_g[\psi_\pm^2] =F_x[P(x)]\,.
\end{equation}
The last step follows from Eq.~(\ref{shifter}), and the fact that in this case the eigenvalues of $\mathbf{A}$ are $\pm1$.
\end{proof}
\subsection{Imaginary weak values}
By equation (\ref{momentumwf}), we see that the joint probability distribution is 
\begin{align}
\mathcal{P}(b,k_x)&=\frac{1}{2}|\tilde{\psi}(k_x)|^2\left[\delta_{b1}\left(1+\cos(\phi_i-\phi_f+2gk)\right)\right.\nonumber\\
&+\left.\delta_{b0}\left(1-\cos(\phi_i-\phi_f+2gk)\right)\right].\\
&=\delta_{b1}\bar{P}^{\textrm{pass}}+\delta_{b0}\bar{P}^{\textrm{fail}}.
\end{align}
We can once more employ Eq.~(\ref{jointfisher}) to arrive at 
\begin{align}
F_g[\mathcal{P}(b,k_x)]=&\frac{1}{2}\int \textrm{d}x\frac{(\tilde{P}(k_x))^2}{\tilde{P}(k_x)}\left[\frac{(-2k\sin(2gk+\phi_i-\phi_f))^2}{1+\cos(2gk+\phi_i-\phi_f)}\right.\nonumber\\
&+\left.\frac{(-2k\sin(2gk+\phi_i-\phi_f))^2}{1-\cos(2gk+\phi_i-\phi_f)}\right]\nonumber\\
=&\int \textrm{d}x|\tilde{\psi}(k_x)|^2 (-2k)^2\nonumber\\
=&4\langle k^2 \rangle\nonumber\\
=&F_x[P(x)].
\end{align}
In the last step we employed equation (\ref{HUP}). The above results agree with the approximate analysis of Section~\ref{failedpostselection}. Of course for qubits $\langle \mathbf{A}^2\rangle = 1$. 

\section{Discussion}
\label{discussion}
In this article, we have studied the metrological advantages of weak-value amplification from a parameter estimation perspective, and derived the ultimate limits on uncertainty that apply. Importantly, our analysis encompasses the presence of jitter and of finite detector resolution. Under our choice of metric, and under these prevalent examples of technical imperfection, we have found there to be no advantage whatsoever for real weak-values. 

For imaginary weak-values we have shown that any advantage -- if one exists at all -- will be menial: due, for example, to an (often incidental) mismatch of noise in position and momentum space rather than an anomalously large mean.
We therefore conclude that the `amplification' aspect of the weak-value formalism offers no fundamental benefits over suitable standard strategies. Further analytical and numerical results suggest that our conclusion remains valid outside the regime of validity of the AAV approximation. 

Previously, we have studied WVA in the context of phase estimation, and found that limitations imposed by decoherence are not mitigated by the technique~\cite{KneeBriggsBenjamin2013}. The combination of our earlier result (concerning quantum noise affecting the wavefunction prior to detection), with the findings of this paper (concerning noise and imperfections relating to the detection apparatus itself) restrict the class of noise where one expects superiority of WVA in metrology scenarios. The very recent work of Ferrie and Combes shows that even errors associated with a finite sample size do not favor a weak-value approach~\cite{FerrieCombes2013}. 

Nonetheless, when circumstances make increased use of resources preferable to complicated post-processing, WVA may remain an attractive option. Our results also do not rule out that aspects other than the amplification itself may still be useful in certain experimental circumstances. Alongside the opportunity to induce changes in the direct meter observable (rather than its conjugate), other possibly useful applications include the reduction in signal intensity; either because this reduces the absolute systematic error (e.g.~due to imperfect optics) or because this avoids detector saturation.
However, these effects are quite distinct from the notion of amplification and hence fall outside the purview of this paper. More theoretical work is required in order that their utility be assessed and quantified. 

\acknowledgments{
We thank Simon Benjamin, Andrew Briggs, \mbox{Yaron Kedem}, Abraham Kofman, Jonathan Leach and Konstantin Bliokh for helpful discussions. This work was supported by the EPSRC [grant number EP/P505666/1], and the National Research Foundation and Ministry of Education, Singapore.
}

\bibliography{gck_full_bibliography_minimal}
\end{document}